\documentclass{amsproc}
\usepackage{amssymb, amsmath}

\newtheorem{theorem}{Theorem}[section]
\newtheorem{lemma}[theorem]{Lemma}
\newtheorem{corollary}{Corollary}
\newtheorem{claim}{Claim}
\theoremstyle{definition}
\newtheorem{definition}[theorem]{Definition}
\newtheorem{example}[theorem]{Example}

\theoremstyle{remark}

\numberwithin{equation}{section}

\begin{document}

\title{On Unique Independence Weighted Graphs}

\author{Farzad Didehvar}
\address{Department of Mathematical and Computer Science, Amirkabir University of Technology, Tehran, Iran.}
\email{didehvar@aut.ac.ir}

\author{Ali D. Mehrabi}
\address{Department of Mathematical Science, Sharif University of Technology, Tehran, Iran.}
\curraddr{Department of Mathematical Science, Yazd University, Yazd,
Iran.} \email{mehrabi@alum.sharif.edu \& mehrabi@yazduni.ac.ir}

\author{Fatemeh Raee B.}
\address{Department of Mathematical Science, Sharif University of Technology, Tehran, Iran.}
\email{f\_raee@alum.sharif.edu}

\date{January 1, 1994 and, in revised form, June 22, 1994.}



\begin{abstract}
An {\sf independent set} in a graph {\it G} is a set of vertices no
two of which are joined by an edge. A vertex-weighted graph
associates a weight with every vertex in the graph. A
vertex-weighted graph {\it G} is called a {\sf unique independence
vertex-weighted graph} if it has a unique independent set with
maximum sum of weights. Although, in this paper we observe that the
problem of recognizing unique independence vertex-weighted graphs is
NP-hard in general and therefore no efficient characterization can
be expected in general; we give, however, some combinatorial
characterizations of unique independence vertex-weighted graphs.
This paper introduces a motivating application of this problem in
the area of {\sf combinatorial auctions}, as well.
\end{abstract}

\maketitle

\section*{Introduction and preliminaries}

In this paper, we focus on graphs whose vertices have real weights
and call such graphs for simplicity, just weighted graphs. Also, we
study unique independent sets in finite vertex weighted graphs. For
the definition of basic concepts and
notations not given here one may refer to a textbook in graph theory, for example \cite{G}, and \cite{I}. \\
Let $G=(V,E)$ be a simple undirected graph with the vertex set
$V=\{1,2,\cdots,n\}$, the edge set $E$ and a nonnegative weight
$w(i)$ associated with each vertex $i\in V$. The weight of
$S\subseteq V(G)$ is defined as $w(S)=\sum_{i\in S} w(i)$. A subset
$I$ of $V(G)$ is called an {\sf independent set} (or a stable set)
if the subgraph $G[I]$ induced by $I$ of $G$ has no edges. A {\sf
maximum weighted independent set}, also called $\alpha$-set, is an
independent set of the largest weight in {\it G}. The weight of a
maximum weighted independent set in {\it G} is denoted by
$\alpha(G)$. A weighted graph {\it G} is a {\sf unique independence
weighted graph}, if
{\it G} has a unique independent set with maximum sum of weights. \\
Characterizing unique independence graphs and various
generalizations of this concept has been a subject of research in
graph theory literature. As a few examples, we refer the interested
reader to \cite{B}, \cite{D}, and \cite{C}. Also, some existing
papers have focused on finding or even approximating the maximum
independent set problem in weighted graphs. See \cite{A}, and
\cite{E} for more details. As we will observe in this paper, this is
not coincidental: we show that the problem of recognizing unique
maximum independence weighted graphs is NP-hard in general and
therefore no efficient characterization of this concept can be
expected in general.\\ To our best knowledge, this is the first
paper discussing the unique maximum weighted independent set problem
and gives some characterizations of it and, most importantly,
defines the problem of {\sf unique combinatorial auctions.} How
weighted graphs and combinatorial auctions are related, is given in
the last section, where we make some future research directions. For
more details on combinatorial auctions, the interested reader is
referred to a detailed article or a textbook
both on combinatorial auctions, respectively in \cite{F} and \cite{J}.\\
The rest of the paper is organized as follows: Section 1 gives a
characterization of unique independence weighted graphs as
generalization of the unique independence graphs. Section 2
introduces some theorems on characterization of unique independence
weighted graphs most of which are based on neighborhood concept. In
section 3 we show the NP-hardness of recognizing the unique
independence weighted graphs and finally section 4 presents some
notes on how much this article would be interesting and what lines
of future researches this article may create.


\section{Unique independence weighted graphs}
\label{UWIG} In this section we exhibit one basic theorem in
addition to a corollary obtained from the theorem, both as
generalizations of unique independence graphs.

\begin{theorem}
\label{T1} Let G be a weighted graph and let I be an $\alpha$-set of
G.
Then the following conditions are equivalent:\\
\\ (i) G is a unique independence weighted graph and I is the unique
$\alpha$-set of G.
\\(ii) For every $x\in I$ we have $\alpha(G\backslash
\{x\})<\alpha(G)$.
\end{theorem}

\begin{proof}
$(i)\Rightarrow(ii)$ Suppose there exist $x\in I$ that $\alpha(G
\backslash\{x\})\geq \alpha(G)=w(I).$ So, $G\backslash\{x\}$
contains an independent set $I'$ which differs from $I$ and also has
the property $w(I')\geq w(I).$ This obviously contradicts either
maximality or the uniqueness of $I$.
\\
\\
$(ii)\Rightarrow (i)$ Suppose $G$ has another maximum weighted
independent set $I', w(I')=w(I)$, and $x\in I\backslash I'$. The set
$I'$ remains a maximum weighted independent set of $G\backslash
\{x\}$. So, $w(I') = \alpha(G\backslash \{x\} )<\alpha(G)=w(I),$
which is a contradiction. So, $G$ has a unique $\alpha$-set.
\end{proof}

\begin{corollary}
Let G be an edge weighted graph and let M be a maximum matching of
G. The following conditions are equivalent:\\
\\
(i) M is a unique maximum matching of G.\\
(ii) For every $e\in M$ we have $\alpha'(G\backslash
\{x\})<\alpha'(G)$.
\end{corollary}

\begin{proof}
Any maximum matching of {\it G} is corresponding to a maximum
independent set of its line graph, {\it L(G)}. So the statement
follows from Theorem \ref{T1}.
\end{proof}

\section{Neighborhood-based characterization of unique independence weighted graphs}
\label{CUWG} In this section, we try to state some theorems on
characterizations of unique weighted graphs mainly based on the
neighborhood concept.

\begin{definition}
For any vertex $x\in V(G)$ the open neighborhood of x in G, N(x,G), is defined as:\\
$N(x,G)=N_G(x)=\{ y\in V(G) | xy\in E(G)\}$. In addition, the
extension of this concept to any subset {\it I} of vertices of a
graph {\it G} is defined as: $N_G(I)=\cup_{x\in I} N_G(x)$.

\end{definition}

\begin{definition}
For a subset I of V(G) and a vertex $x\in I$, we define:\\
$$p_I^G(x)=N_G(x)\backslash N_G(I\backslash\{x\})$$
\\
Furthermore, for every subset I of V(G), we define the set p$_G(I)$
as: $$p_G(I) = \bigcup_{x\in I} p_{I}^G(x)$$
\end{definition}

The following lemma gives a sufficient neighborhood-oriented
condition by which the uniqueness of a weighted graph is
established.
\begin{lemma}
\label{L1} Let G be a weighted graph and let $I$ be an $\alpha$-set
of G. If for any $I_0\subseteq I$ we have $w(p_G(I_0))<w(I_0)$, then
G is unique independence weighted graph and $I$ is the unique
$\alpha$-set of G.
\end{lemma}

\begin{proof}
By contradiction. Let $I'$ be another $\alpha$-set of {\it G}. This
means, $I\backslash I' \neq \emptyset$, $w(I)=w(I')$, and also
$w(I\backslash I')=w(I' \backslash I)$. Clearly, $I'\backslash
I\subseteq p_G(I\backslash I')$ and thus $w(I'\backslash I)\leq
w(p_G(I\backslash I'))$. Replacing the left side of this equation by
its equivalent, $w(I\backslash I')$ , results $w(I\backslash I')\leq
w(p_G(I\backslash I'))$. Now, taking $I_0=I\backslash I'$
contradicts the hypothesis of the lemma.
\end{proof}
The converse of Lemma \ref{L1} is not true for all weighted graphs.
The following example gives a counter-example.

\begin{example} Suppose that G is the following
weighted graph.
\begin{center}
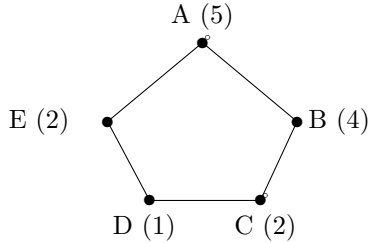
\begin{figure}[ht]
\label{counter-example}

\def\emline#1#2#3#4#5#6{%
\put(#1,#2){\special{em:moveto}}%
\put(#4,#5){\special{em:lineto}}}

\def\newpic#1{}
\unitlength 0.7mm \special{em:linewidth 0.4pt} \linethickness{0.4pt}
\begin{picture}(45,45)(110,110) 
%
\put(126,150){\r{\circle*{2}}} \put(108,135){\circle*{2}}
\put(144,135){\circle*{2}} \put(116,120){\circle*{2}}
\put(137,120){\r{\circle*{2}}}
\emline{126}{150}{1}{108}{135}{2}
\emline{126}{150}{1}{144}{135}{2}
\emline{108}{135}{1}{116}{120}{2}
\emline{144}{135}{1}{137}{120}{2}
\emline{116}{120}{1}{137}{120}{2}
\put(95,135){\makebox(0, 0)[cc]{E (2)}} \put(115,115){\makebox(0,
0)[cc]{D (1)}} \put(138,115){\makebox(0, 0)[cc]{C (2)}}
\put(152,135){\makebox(0, 0)[cc]{B (4)}} \put(126,155){\makebox(0,
0)[cc]{A (5)}}

\end{picture}
\vspace*{-3mm} \caption{A Counter-Example}

\end{figure}
\end{center}

The numbers enclosed in parentheses, are the vertices' weights.
Suppose $I=\{A,C\}.$ $I$ is a unique $\alpha$-set of G. If
$I_0$=\{A\}, then $p_G(I_0)=\{E\}$. So, $w(p_G(I_0))=2<5=w(A)$. If
$I_0$=\{C\} then $p_G(I_0)$=\{D\}. So, $w(p_G(I_0))=1<2=w(C)$. If
$I_0$=\{A,C\}, then $p(I_0)=\{B,E,D\}$. So, $w(p_G(I_0))=7 \nless
7$. Therefore, there is a subset of $I$ not satisfying the condition
given by the Lemma \ref{L1}.

\end{example}

The next theorem exhibits the fact that the converse of Lemma
\ref{L1} is true for all trees.

\begin{theorem}
\label{T2} Let T be a weighted tree and let $I$ be an $\alpha$-set
of T. The following conditions are equivalent:\\
\\
(i) T is unique independence weighted tree and $I$ is the unique
$\alpha$-set of T. \\
(ii) For every $I_0\subseteq I$, we have $w(p_T(I_0))<w(I_0)$.
\end{theorem}

\begin{proof}
(ii) $\Rightarrow$ (i) Is implied directly from Lemma \ref{L1}.\\
\\
(i) $\Rightarrow$ (ii) The proof by contradiction. Let $A = \{I' \ |
\ I'\subseteq I \ and \ w(p_T(I'))\geq w(I')\}$ and $W_A = \{ w(I')
| I'\in A \}$. Based on the contrary hypothesis, the set $W_A$ is
not empty and therefore making the assumption that $\alpha$ to be
the smallest element of $W_A$ is allowed. Let us to suppose that the
set $I_0$ be that member set of $A$ corresponding to the value
$\alpha$, $\alpha= w(I_0)$. Take $I'=(I\backslash I_0)\bigcup
p_T(I_0)$ and claim that $I'$ is an independent set so that
$w(I')\geq w(I)$. In order to prove this claim, it suffices to show
that $p_T(I_0)$ is an independent set. We show this by
contradiction. Let $x,y\in p_T(I_0)$ and $xy\in E(T)$. {\it T} is a
tree, so $xy$ is a bridge of {\it T}. Therefore $T\backslash \{xy\}$
has two components, say $T_1$ and $T_2$. Suppose $I_1=I_0\bigcap
V(T_1)$ and $I_2=I_0\bigcap V(T_2)$. Take $x,y\in p_T(I_0)$, so
there are $x',y'\in I_0$ that $x\in p_T(x')$ and $y\in p_T(y')$.
Thus $I_1,I_2\neq \emptyset$ and $I_1\bigcap I_2=\emptyset$ and
$I_1\bigcup I_2=I_0$. On the other hand, $p_T(I_1)\bigcap
p_T(I_2)=\emptyset$ and $p_T(I_1)\bigcup p_T(I_2)=p_T(I_0)$. So:

\begin{equation}
\label{E1} w(p_T(I_0))=w(p_T(I_1))+w(p_T(I_2))
\end{equation}

But we have $I_1 \subsetneqq I_0$ and $I_2 \subsetneqq I_0$. By
minimality of $w(I_0)$ we have $w(p_T(I_1))<w(I_1)$ and
$w(p_T(I_2))<w(I_2)$. This contradicts Equation \ref{E1}.
\end{proof}

The following theorem gives a general condition under which the
uniqueness of a weighted graph is established. Before proceeding
this theorem, we make a prerequisite definition.

\begin{definition}
For any $I\subseteq V(G)$, we denote the maximum weighted
independent set of $p_G(I)$ by $m(I)$.
\end{definition}

\begin{theorem}
\label{T3} Let G be a weighted graph and $I$ be an $\alpha$-set of
G. The following conditions are
equivalent:\\
\\
(i) G is unique independence weighted graph and $I$ is the unique
$\alpha$-set of G.\\ (ii) For every $I_0\subseteq I$, we have
$w(m(I_0))<w(I_0)$.
\end{theorem}

\begin{proof}
(i) $\Rightarrow$ (ii) This part is done by contradiction. Suppose
there exist $I_0\subseteq I$ such that $w(m(I_0))\geq w(I_0)$. Let
$I'=(I\backslash I_0)\bigcup m(I_0)$. So $I'\neq I$  and $w(I')\geq
w(I)$. On the other hand, $m(I_0)$ and $I\backslash I_0$ are
independent sets, $m(I_0)\subseteq p_G(I_0)$ and $p_G(I_0)\bigcap
N(I\backslash
I_0)=\emptyset$. Hence, $I'$ is an independent set which is a contradiction.\\
\\
(ii) $\Rightarrow$ (i) This part is also done by contradiction.
Suppose $I'$ be another $\alpha$-set of {\it G}. This means
$w(I)=w(I')$ and also:
\begin{equation}
\label{E2} w(I\backslash I')=w(I'\backslash I)
\end{equation}
In addition, $I'\backslash I\subseteq p_G(I\backslash I')$ and
$I'\backslash I$ is an independent set. Therefore, $w((I'\backslash
I))\leq w(m(I\backslash I'))$. Now, let $I_0=I\backslash I'$ and
thus, by condition (ii), we have $w(m(I\backslash I'))<w(I\backslash
I')$. So, we obtanied $w(I'\backslash I)\leq w(m(I\backslash
I'))<w(I\backslash I')$, which contradicts Equation \ref{E2}.
\end{proof}

The prior theorem needs to see whether all subsets of the
$\alpha$-set, have the given condition by the theorem. If so, then
the uniqueness will be confirmed. What if not? Hence, it seems that
Theorem \ref{T3} for those weighted graphs whose independent sets
are decent large in size, needs a time-consuming process before
making any decision. For the sake of this, we exhibit the theorem
below, particularly useful and thus important for those weighted
graphs including a decent large independent set.

\begin{theorem}
\label{T4} Let G be a weighted graph and let I be an $\alpha$-set
of G. The following conditions are equivalent:\\
\\
(i) G is unique independence weighted graph and I is the unique
$\alpha$-set of G. \\
(ii) For every nonempty independent subset J of V(G)$\backslash$I,
we have: $w(N(J)\cap I)>w(J)$.
\end{theorem}
\begin{proof}
(i) $\Rightarrow$ (ii) If $J$ is a nonempty independent subset of
$V(G)\backslash I$, then $(I\backslash N(J))\cup J$ is an
independent set in $G$. $I$ is the unique independent set of $G$, so
$w((I\backslash N(J)))<w(I)$. Thus we have: $w(N(J)\cap I)>w(J)$.\\
\\
(ii) $\Rightarrow$ (i) Let $I'$ be an independent subset of {\it G}.
It suffices to show that $w(I')<w(I)$. Since $I'\backslash I$ is a
nonempty independent subset of $V(G)\backslash I$, we have:
$w(N(I'\backslash I)\cap I)>w(I'\backslash I)$. Moreover,
$N(I'\backslash I)\cap I\subseteq I\backslash I'$ and therefore
$w(I')=w(I'\cap I)+w(I'\backslash I)<w(I'\cap I)+w(N(I'\backslash
I)\cap I)\leq w(I\cap I')+w(I\backslash I')=w(I)$. So, $w(I')<w(I)$.
\end{proof}

\par\noindent
The next theorem is intended to prove this our conjecture\footnote{A
conjecture that we had made at the early steps of this work.} on
whether or not a given weighted graph has a unique makeup? In fact,
the following theorem proves that for every weighted graph, there
are many other weighted graphs whose independent sets are same as
the given graph, if their vertices' weights have been drawn from a
predefined real interval connected to the vertices' weights of the
original given graph. This theorem proves this and provides that
interval, as well.
\begin{theorem}
\label{T5}
 Let G be a weighted graph and $I$ be an $\alpha$-set of
G. If G is a unique independence graph and $I$  is the unique
$\alpha$-set of G, then there is a positive real number,
$\epsilon>0$, such that if the weights of G's vertices change in
$(w(x)-\epsilon, w(x)+\epsilon)$, then G with these new weights
remains a unique independence weighted graph with the same
$\alpha$-set.
\end{theorem}
\begin{proof}
$I$ is the unique $\alpha$-set of $\it G$ so by Theorem \ref{T3} for
every $I_0\subseteq I$, $w(m(I_0))<w(I_0).$ \\Let:
\begin{itemize}
\item $\sigma=\min$ \{ $w(I_0)-w(m(I_0)) | I_0\subseteq I$ \},
\item $\eta=\min$ \{ $w(I)-w(I_0) | I_0$ is an independent set of $G$
\},
\item $\nu=\min$ \{ $w(m(I_0))-w(J) | I_0\subseteq I$ and $J$ is an independent set of $p_G(I_0)$
\}.
\end{itemize}
and let $\delta=\min$ \{ $\sigma, \eta, \nu$ \} and also
$\epsilon=\frac{\delta}{n+1}$, where $n$ is the number of $\it G$'s
vertices. Suppose $G'$ is a copy of $G$, with new changed vertices
weights, $w'$, such that for every $x\in V(G)$:
$w(x)-\epsilon<w'(x)<w(x)+\epsilon$. Now, we make the following
claim in order to complete the proof.
\begin{claim}
\label{C2} $G'$ is a unique independence weighted graph and $I$ is
the unique $\alpha$-set of $G'$.
\end{claim}
{\bf Proof of Claim 1:} By definition of $\eta$, $I$ is an
$\alpha$-set of $G'$. To prove the uniqueness of $I$, it's
sufficient to show that for every $I_0\subseteq I$,
$w'(m(I_0))<w'(I_0)$. Proof by contradiction: Suppose there is a
subset of vertices like $J$, such that $J\subseteq I$ and also:
\begin{equation}
\label{E3} w'(m(J))\geq w'(J)
\end{equation}
By definition of $\nu$, $m(J)$ is an $\alpha$-set of $p_G(J)$ in
both $G$ and $G'$. So we have:
\begin{equation}
\label{E4} w(J)-|J|.\epsilon \leq w'(J) \leq w(J)+|J|.\epsilon
\end{equation}
\begin{equation}
\label{E5} w(m(J))-|m(J)|.\epsilon \leq w'(m(J)) \leq
w(m(J))+|m(J)|.\epsilon.
\end{equation}
By combining Equations \ref{E3}, \ref{E4}, \ref{E5} and some simple computations, we achieve:\\
\begin{equation}
\label{E6} w(J)\leq w'(J)+|J|.\epsilon\leq w'(m(J))+|J|.\epsilon\leq
w(m(J))+|m(J)|.\epsilon+|J|.\epsilon.
\end{equation}
From Equation \ref{E6}, the following equation is obtained.
\begin{equation}
\label{E7} w(J)\leq w(m(J))+\epsilon.(|J|)+|m(J)|\leq
w(m(J))+\epsilon.n<w(m(J))+\delta.
\end{equation}
Finally, we achieve: $w(J)-w(m(J))<\sigma$. Obviously, this
contradicts the definition of $\sigma$. This completes the proof of
Claim \ref{C2} and thus the proof of Theorem \ref{T5} is now
complete.
\end{proof}

\begin{corollary}
For every given weighted graph $G$, there are infinite number of
weighted graphs whose vertices' weights are real and their
independent sets are the same as $G$.
\end{corollary}


\section{Complexity of unique maximum weighted independent set problem}
We prove that the following problems are NP-hard. Both problems ask
for detecting whether a given vertex weighted graph has a unique
maximum weighted independent set; in the first problem, the input
contains a candidate for the unique maximum weighted independent set
in addition
to the graph.\\

Problem $UI_1:$
\\
{\bf Input:} A weighted graph {\it G}, a set {\it I} of the vertices of {\it G}.\\
{\bf Question:} Is {\it I} the unique maximum weighted independent set in {\it G}?\\

Problem $UI_2:$
\\
{\bf Input:} A weighted graph {\it G}.\\
{\bf Question:} Does {\it G} have a unique maximum weighted independent set?\\

We prove the NP-hardness of these problems by reducing the following problem to them:\\

Problem {\large W}{EIGHTED} {\large I}{NDEPENDENT} {\large S}{ET}:\\
{\bf Input:} A weighted graph {\it G}, an integer {\it k}.\\
{\bf Question:} Does {\it G} contain an independent set of weight at least {\it k}?\\
\\
The latest problem is NP-Complete and one may refer to \cite{H}
for a proof.\\

Now, the following two theorems exhibit the complexity classes to
which the problems $UI_1$ and $UI_2$ are belonging.

\begin{theorem}
Problem $UI_1$ is coNP-complete.
\end{theorem}
\begin{proof}
First, we show that this problem is in coNP. To see this, it is
enough to observe that a witness for the non-membership of an
instance $(G,I)$ in $UI_1$ is an independent set of weight greater
than or equal to $w(I)$. We now show that the problem is
coNP-complete by showing a reduction from the complement of {\large
W}{EIGHTED} {\large I}{NDEPENDENT} {\large S}{ET} problem to this
problem. Given an instance $(G,k)$ of {\large W}{EIGHTED} {\large
I}{NDEPENDENT} {\large S}{ET}, construct a graph $H$ by adding $k$
vertices to $G$ and all the edges between these $k$ vertices and the
vertices of $G$ (but no edge between the $k$ new vertices) and then
set the weight of each new added vertex by 1. Let $H$ denote the
resulting graph, and $I$ denote the set of $k$ vertices in $V(H)
\backslash V(G)$. We claim that $(G,k)\in$ {\large W}{EIGHTED}
{\large I}{NDEPENDENT} {\large S}{ET} if and only if $(H,I)\notin
UI_1$. This is because by construction, every independent set of $H$
is either a subset of $I$, or an independent set in $G$. Therefore,
$I$ is the unique maximum weighted independent set in $H$ if and
only if $G$ does not contain an independent set whose weight exceeds
$w(I)$. Therefore, the above construction is a polynomial time
reduction from the complement of {\large W}{EIGHTED} {\large
I}{NDEPENDENT} {\large S}{ET} to $UI_1$. This completes the proof of
coNP-completeness of $UI_1$.
\end{proof}

For problem $UI_2$, the situation is less clear, as the problem does
not seem to be in NP or coNP. It is not difficult to show that this
problem is in the complexity class $\sum_{2}$, but we do not know if
it is $\sum_{2}$-complete. However, we can still show that the
problem is intractable, assuming $P\neq NP$.

\begin{theorem}
Problem $UI_2$ is NP-hard.
\end{theorem}
\begin{proof}
As in the proof of the previous theorem, we show a reduction from
the complement of {\large W}{EIGHTED} {\large I}{NDEPENDENT} {\large
S}{ET} to $UI_2$. Given an instance $(G,k)$ of {\large W}{EIGHTED}
{\large I}{NDEPENDENT} {\large S}{ET}, construct a graph $H$ by
adding a set $I$ of $k+1$ vertices and another set $R$ of two
vertices with 1 as the weight of each vertex in both sets, to $G$.
The edges of $H$ are the edges of $G$ plus edges between all
vertices in $I$ and all vertices in $V(G)\cup R$, and also one edge
between the two vertices of $R$. We claim that $(G,k)\in$ {\large
W}{EIGHTED} {\large I}{NDEPENDENT} {\large S}{ET} if and only if
$H\notin UI_2$. This is because by construction, every weighted
independent set of $H$ is either a subset of $I$, or a subset of
$V(G)\cup R$. The weight of largest independent set in $V(G)\cup R$
is precisely $\alpha(G)+1$. Therefore, the weight of the largest
independent set of $H$ is $\max(k,\alpha(G))+1$. Therefore, if $G$
has an independent set of weight at least $k$, at least two
$\alpha$-sets in $H$ can be obtained by adding either of the
vertices of $R$ to a maximum weighted independent set of $G$. Thus,
$H\notin UI_2$ in this case. Conversely, if $G$ has no weighted
independent set of weight $k$ or more the unique $\alpha$-set of $H$
is $I$. Therefore, the above construction is a polynomial time
reduction from the complement of {\large W}{EIGHTED} {\large
I}{NDEPENDENT} {\large S}{ET} to $UI_2$. This reduction completes
the proof.
\end{proof}

\section{Concluding remarks and future research directions}

To our knowledge, this is the first paper beginning this line of
research due to the diverse applications this problem have. Of
which, let us to outline one the most important ones which has
recently received the attention of researchers and has also opened
many doors of research opportunities in fields like algorithmic game
theory, computational economics and e-commerce. Combinatorial
auctions are mechanisms for allocating a set of items between a set
of (likely selfish) agents. It's a well-known principle that every
combinatorial auction can be equivalently viewed as a
vertex-weighted graph where bid sets and the set of winners of the
former correspond to the vertices and the maximum independent set of
the latter\cite{E}. Hence, the most significant result of this study
is to define the problem of {\sf unique combinatorial auctions},
those combinatorial auctions whose the set of winners is unique, and
providing some good starting points for this line of research. So,
the authors believe that this article may be considered as the first
step to define the problem of unique combinatorial auctions and also
as the first step to characterize and classify unique combinatorial
auctions not only as a specified problem but also as the first for
future related studies.


\section*{Acknowledgments}
The authors would like to warmly thank Prof. Ebadollah S. Mahmoodian
for providing advices and encouragements. They also take the
opportunity to thank the Institute for studies in applied Physics
and Mathematics (IPM), and particularly Prof. Shahriar Shahriari,
for their warm hospitality and invaluable helps, when they were
doing this study. Also, the second and third authors appreciate the
hospitality of the Department of Mathematical and Computer Science
at the Amirkabir University of Technology, while doing this
research.

\bibliographystyle{amsalpha}

\end{document}